\documentclass[conference]{IEEEtran}
\usepackage{cite}

\ifCLASSINFOpdf
   \usepackage[pdftex]{graphicx}
\else
   \usepackage[dvips]{graphicx}
\fi
\usepackage{url}

\usepackage[utf8]{inputenc}
\usepackage[english]{babel}
\usepackage[T1]{fontenc}
\usepackage{booktabs} 
\usepackage[]{algorithm2e}
\usepackage{color}
\usepackage{array}
\usepackage{amsmath}

\usepackage{balance}
\usepackage{amsthm}
\usepackage{xcolor}
\usepackage{enumitem}
\usepackage{amssymb}
\usepackage{color, colortbl}
\definecolor{darkred}{rgb}{0.75,0,0}

\newtheorem{theorem}{Theorem}[section]

\begin{document}
%
\newcommand{\name}{FileBounty\xspace}
\title{\name: Fair Data Exchange}

\author{%
  \IEEEauthorblockN{%
    Simon Janin\IEEEauthorrefmark{2}\IEEEauthorrefmark{1},
    Kaihua Qin\IEEEauthorrefmark{3}\IEEEauthorrefmark{1},
    Akaki Mamageishvili\IEEEauthorrefmark{4} and
    Arthur Gervais\IEEEauthorrefmark{3}%
  }%
  \IEEEauthorblockA{\IEEEauthorrefmark{2} X80 Security}%
  \IEEEauthorblockA{\IEEEauthorrefmark{3} Imperial College London}%
  \IEEEauthorblockA{\IEEEauthorrefmark{4} ETH Zurich}%
}
\maketitle
\begingroup\renewcommand\thefootnote{\IEEEauthorrefmark{1}}
\footnotetext{These authors contributed equally to this work.}
\endgroup
\thispagestyle{plain}
\pagestyle{plain}

\begin{abstract}
Digital contents are typically sold online through centralized and custodian marketplaces, which requires the trading partners to trust a central entity. We present \name, a fair protocol which, assuming the cryptographic hash of the file of interest is known to the buyer, is trust-free and lets a buyer purchase data for a previously agreed monetary amount, while guaranteeing the integrity of the contents. To prevent misbehavior, \name guarantees that any deviation from the expected participants' behavior results in a negative financial payoff; i.e.\ we show that honest behavior corresponds to a subgame perfect Nash equilibrium. Our novel deposit refunding scheme is resistant to extortion attacks under rational adversaries. If buyer and seller behave honestly, \name's execution requires only three on-chain transactions, while the actual data is exchanged off-chain in an efficient and privacy-preserving manner. We moreover show how \name enables a flexible peer-to-peer setting where multiple parties fairly sell a file to a buyer.
\end{abstract}


%
\IEEEpeerreviewmaketitle

\section{Introduction}
Fair multi-party exchange protocols have been extensively studied~\cite{kilincc2015optimally,garay2004efficient, gauthierdickey2014secure, bao1998efficient, asokan1998optimistic, dodis2003breaking, park2003constructing, zhou2000some, garay2003timed}. In light of the emergence of permissionless blockchains (e.g.\ Bitcoin), however, the design of fair exchange protocols has experienced a renaissance~\cite{banasik2016efficient,bentov2014use,bartoletti2016constant,andrychowicz2014fair, dziembowski2018fairswap}. The blockchain's properties allow to construct marketplaces without reliance on a centralized trusted third party (TTP), and instead enable the automated dispute mediation and integrated monetary transfers among buyers and sellers.

In this work, we first consider the following problem: we are given two participants, a seller and a buyer that are willing to exchange a digital file for a monetary amount. We assume that both parties mutually do not trust each other, and we therefore require a fair data exchange protocol. Our definition of fairness implies that a dishonest participant always suffers negative payoffs. In a second step, we consider the problem of a peer-to-peer (P2P) file exchange protocol similar to Bittorrent~\cite{cohen2008bittorrent}, whereby a buyer intents to purchase a file from multiple sellers which are expect to be rewarded proportionally to their provided contents.

We propose \name, a fair protocol that allows a buyer to purchase a digital file from a seller for a monetary amount. \name relies on the assumption that the buyer knows the cryptographic hash\footnote{Note in the extended P2P exchange protocol, we term the identity of a file \textit{fingerprint}.} of the file of interest (for example acquired through a semi-trusted third party, such as VirusTotal or a Torrent website\footnote{Note that we consider the semi-trusted third party offline, because it is not required to interact during or after the file exchange.}). Given the file hash, a buyer e.g.\ publishes on a blockchain an openly available bounty to purchase the file of interest from a willing seller. Alternatively, a seller can advertise its willingness to sell a file, similar to Bittorrent magnet links, connect with a buyer and complete the file and value transfer successfully. The file transfer is executed off-chain, i.e.\ the actual file content is not written on the bockchain; thus both protecting the privacy of the file and remaining scalable to any filesize. Even in the case of disputes, no file content needs to be disclosed as our on-chain dispute process mediates via zkSNARKs. The file is moreover exchanged gradually and integrity protected, i.e.\ for each file chunk, which consists of multiple file blocks, transferred, the buyer can efficiently verify that a chunk is a valid chunk of the file of interest. We further extend \name to a distributed P2P setting, whereby a buyer can purchase the file of interest from multiple sellers simultaneously, allowing the transfer speed to scale proportionally to the number of sellers and their network connectivity.

We adopt a rational adversarial model, under which, for each transferred chunk, the seller is guaranteed to receive a proportional payoff, as long as he doesn't become unresponsive (e.g.\ times-out). The smart contract verifies the proper execution of the protocol to assure the correct participant payoffs and automatically handles dispute mediation. We introduce a deposit refunding scheme that guarantees that \name is secure against a rational adversary, i.e.\ the adversary will only follow the honest protocol which corresponds to the subgame-perfect Nash equilibrium. In other words, any deviation for any player (buyer, seller) from the honest protocol behavior will result in a financial loss to this player.
For our design, we make use of the chaining property of Merkle-Damg\aa rd-based hash functions (cf.\ MD5~\cite{rivest1992md5}, SHA1~\cite{eastlake2001us}, SHA2~\cite{standard2002fips}, etc.)

Potential application opportunities of \name are extensive: 
(i) decentralized digital marketplaces,
(ii) purchase of sparsely seeded torrent files, (iii) decentralized file backups, 
(iv) brute-forcing of password hashes given their respective salts and 
(v) exchange of malware samples\footnote{i.e.\ \url{http://www.virustotal.com}}.

As a summary our contributions are as follows:
\begin{description}
    
	\item[Fair one-to-one Exchange] We propose \name{}, a fair data exchange protocol which allows a buyer (who knows the cryptographic hash of the file of interest), to make a publicly available bounty call for a digital file.
	
	\item[Fair Multi-Party Exchange] We present a novel multi-party exchange protocol based on BLS signatures, whereby a buyer can fairly purchase a file from multiple sellers.
	
	\item[Efficient Integrity Verification] By taking advantage of the Merkle-Damg\aa rd construction (e.g., in SHA256), our design allows the buyer (who only knows the file hash) to verify for each received file chunk, that the file chunk belongs to the file of interest.

	\item[Privacy-Preserving] We conduct the on-chain dispute mediation in a privacy-preserving manner through zkSNARKs.
	
	\item[Nash Equilibrium] We show that \name is secure under a rational adversary, i.e.~we show that honest behavior is the subgame perfect Nash equilibrium. Misbehavior of any of the involved parties is detected while the data is gradually exchanged. Contrary to previous schemes, ours is resistant to extortion attacks under the rational adversary.
\end{description}

The remainder of the paper is organized as follows. Section~\ref{sec:background} covers the background, Section~\ref{sec:design} presents the system and adversarial model and
Section~\ref{sec:datails} details \name{}'s design. Section~\ref{sec:securityanalysis} focuses on the security analysis while we present our evaluation in Section~\ref{sec:evaluation}. We overview related work in Section~\ref{sec:relatedwork} and conclude the paper in Section~\ref{sec:conclusion}.

\section{Background}\label{sec:background}
In this section, we discuss the required background.
\begin{description}
\item[Blockchain]
    In permissionless blockchains, any peer is free to join or leave the network at any moment. Proof-of-Work (PoW), introduced by ~\cite{dwork1992pricing}, is a computationally expensive puzzle that is the currently widest deployed mechanism to operate permissionless blockchains (e.g.\ Bitcoin). Besides the simple transfer of monetary value, transactions can execute user-defined programs, commonly referred to as smart contracts, which alter the state of the blockchain. For an in-depth background, we refer the reader to~\cite{tschorsch2016bitcoin}.

\item[Game Theory and Cryptoeconomics] Fair multi-party protocols can be modeled as non-cooperative games, where each party acts rationally in its own best interest. The subgame perfect Nash equilibrium~\cite{osborne1994course} is a game theoretic solution concept of the extensive form game that we use in this paper to model the sequential interaction between the seller and the buyer. In this solution, each player plays the optimal strategy at each possible state. If any player would change its strategy in this solution, the expected financial gain is a non-positive number, as such, no player has an incentive to change its strategy in any state.
    
    
\item[Merkle-Damg\aa rd and Sponge Construction]
    The Merkle-Damg\aa rd construction~\cite{goldwasser1996lecture} allows to build a collision-resistant cryptographic hash function with an arbitrary size input from a collision-resistance one-way compression function. The arbitrary-length input data is divided into fixed-length data chunks and use an underlying cryptographic hash function in order to produce a final hash (cf. Figure~\ref{fig:merkledamgard}). For security purposes, the length of the actual input data needs to be appended at the end of the data. There are alternatives to the Merkle-Damg\aa rd construction, like for example the Sponge construction which is the basis for the SHA-3 hash function~\cite{SHA3}.
    
\item[Boneh-Lynn-Shacham signature scheme] The Boneh-Lynn-Shacham (BLS) signature scheme~\cite{boneh2001short} is a short digital signature technique using a bilinear pairing on a curve which supports signature aggregation~\cite{boneh2003aggregate} (e.g. compressing multiple signatures signed by different private keys on different messages in a single signature).

\item[Zero Knowledge Proofs] Zero-knowledge proofs (ZKP)~\cite{feige1988zero} allow a prover to prove to another verifier that a statement is true but without revealing any information of it. zkSNARKs~\cite{bitansky2012extractable} are succinct ZKP methods in which the size of the proof is small enough to be verified efficiently (e.g.~by a smart contract). zkSNARKs are moreover non-interactive i.e.\ no interaction is required between the prover and the verifier.
\end{description}
    
\begin{figure}[!htb]
\centering
\includegraphics[width=0.2\textwidth]{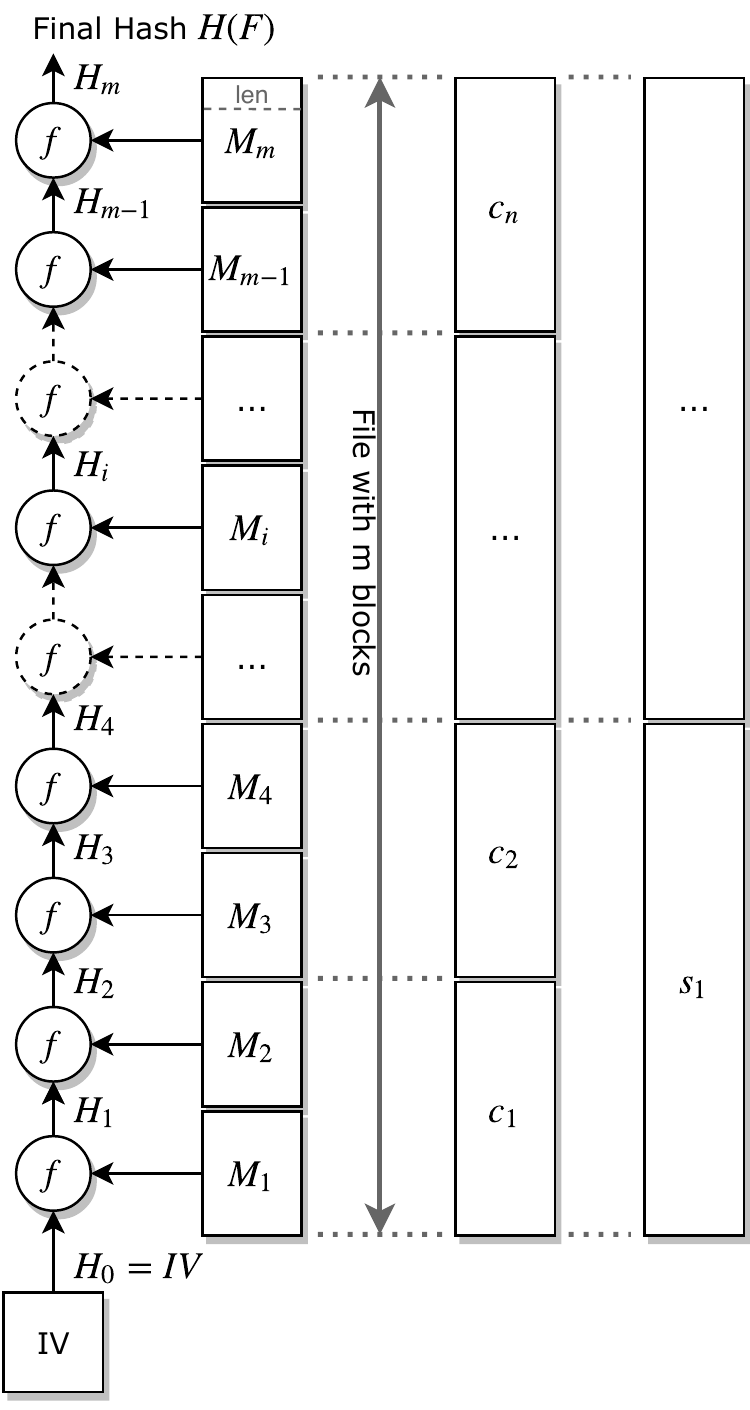}
\caption{Merkle-Damg\aa rd construction. The input file is divided into $m$ blocks, and the construction is initialized with a fixed initialization vector (IV). Intermediate states ($H_i$) represent temporary hashes, while the final file hash is output at the top. Note that the length of the file ($len$) is attached to the end of the file. Note that $H(F)$ corresponds to $H_m$.}
\label{fig:merkledamgard}
\end{figure}

\section{\name: Overview}\label{sec:design}
In this Section, we first outline the system and adversarial model followed by the detailed protocol specification.

\subsection{System Model}
Four parties are involved in our system: (i) a seller, (ii) buyer (iii) semi-trusted third party which provides the cryptographic hash of the file of interest and (iv) a smart contract.
\begin{description}
    \item[Permissionless Blockchain] We assume the existence of a permissionless (e.g.\ Proof-of-Work) blockchain that supports the convenient deployment of smart contracts. Participants have identities in the form of public keys that are tied to their respective private keys ($Pub_x/Priv_x$) to perform e.g.\ digital signatures.
    \item[Buyer] The buyer $B$ is the entity which is willing to purchase the file $F$ with cryptographic hash $H(F)$ from a seller for a given price $F_p$. 
    \item[Seller] The seller $S$ is the entity which claims to own and is willing to transfer the digital file $F$, corresponding to the cryptographic hash $H(F)$, to a buyer.
    \item[Semi-trusted third party] The semi-trusted third party (STTP) is a neutral party that has no incentive to report incorrect hashes. The STTP plays the role of indirect coordinator, providing the cryptographic hashes of files to buyers.
    
    \end{description}

    To exchange the file $F$, we assume a communication protocol external to the blockchain, which requires the buyer or the seller to directly connect (e.g.\ over WebRTC). 
    
    \begin{description}
    \item[Iterated Hash Functions] Our protocol relies on iterated hash functions~\cite{lucks2004design} (such as SHA256). An iterated hash function takes a message $M \in \{0,1\}^{*}$ of any length to compute an s-bit output H(M). For an iterated hash, we split the padded message into $m$ fixed-sized parts $M_1, M_2, ..., M_m \in \{0, 1\}^{l}$. An iterated hash $H$ iterates an underlying compression function $f$, and the final hash is given by $f(f(... f(f(H_0, M_1), M_2) ...), M_m)$, where $H_0$ is some constant initial value usually denoted as Initialization Vector ($IV$).
\end{description}
We term the fixed-sized parts of a file fed into the integrated hash functions \textit{blocks}. Adjacent file blocks are bundled to \textit{chunks} as the transmission unit. In the P2P scheme, files are split into \textit{segments} which are exchanged separately. Segments are exclusive parts of a file containing several file chunks. Figure~\ref{fig:merkledamgard} illustrates the relationship between the three terms.

\subsection{High-Level Operation}

We outline the high-level operations of \name{} in Figure~\ref{fig:overview}, which is divided into four phases.

\begin{description}
    \item[Setup Phase (a)] 
    The sellers start to serve a file $F$ by submitting a set of requisite parameters and the deposit $D_S$ to the smart contract. The buyer who is interested in $F$ can pay the file price $F_P$ with the deposit $D_B$ and initiate off-chain communication channels to the alive sellers.
    \item[File Transfer Phase (b)] 
    Once the on-chain transactions confirmed in the setup phase and the off-chain channels are established, the protocol enters into the file exchange phase. The buyer receives $F$ in chunks from one or multiple sellers. For each transferred chunk, the client verifies the chunk's integrity given the file hash $H(F)$ and the seller receives a signed acknowledgment from the buyer.
    \item[Dispute Phase (c)] {If one participant misbehaves, the aggrieved side is able to claim restitution via the smart contract. If both sides behave honestly until the end of transmission, the protocol will distribute payoff as expected.}
    \item[Payoff Phase (d)] {In the payoff phase, the smart contract allocates the file price and the deposits to the buyer and the seller according to the execution of the protocol.}
\end{description}

\begin{figure}[!htb]
\centering
\includegraphics[width=0.8\columnwidth]{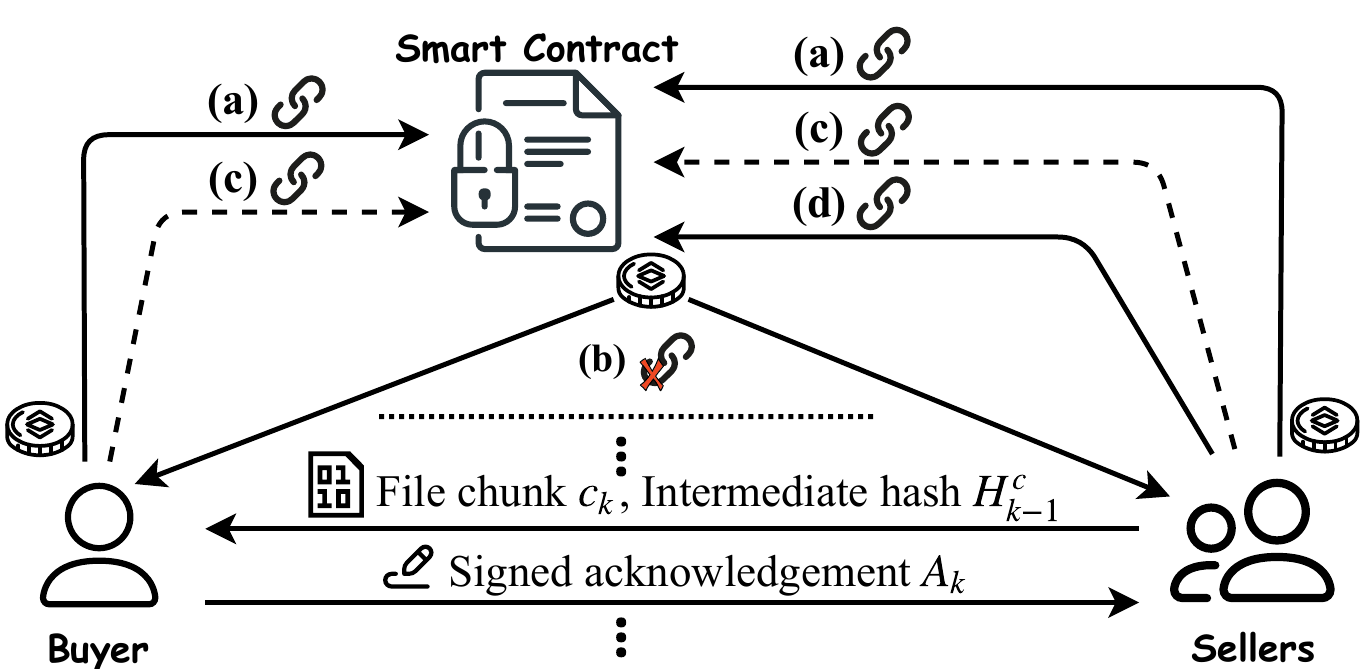}
\caption{Four phases of  \name{} --- (a) setup phase; (b) file exchange phase; (c) dispute phase; (d) payoff phase.}
\label{fig:overview}
\end{figure}

\subsection{Main Properties}
\begin{description}
    \item[Permissionless Blockchain] Blockchain transactions enable a certain transaction throughput and bear a cost borne by the buyer and seller.
    \item[Buyer] The buyer is able to establish a secure communication channel to the seller and interacts with a blockchain-based smart contract. Prior to the exchange, the buyer provides a monetary deposit of value $D_B$. The file of interest is worth $F_v$ to the buyer.
    \item[Seller] The seller is able to establish a secure communication channel to the buyer and interacts with the blockchain-based smart contract. Prior to the file exchange, the seller provides a monetary deposit of value $D_S$. The seller's costs for selling the file amounts to $F_c$.
    \item[Semi-trusted third party] The STTP only assists in the trade and is not involved in the exchange of any monetary nor data transfer between the buyer and the seller.
\end{description}

\subsection{Attacker Models}
\label{sec:attackermodels}
We assume that the underlying blockchain is reasonably resistant against 51\% and double-spending attacks. We assume that the consensus mechanism of the blockchain achieves eventual consistency and that the transactions are non-repudiable and non-malleable. We also assume, that \name{}'s smart contract implementation is secure~\cite{luu2016making}. Additionally, we rely on the correct design and implementation of the cryptographic primitives of the blockchain. We consider the off-chain connection between the seller and buyer to be secure and stable and not controllable by the adversary. All protocol participants are assumed to be able to monitor the state changes of the smart contract and respond timely. We assume the STTP is not colluding with the seller. We also assume a rational adversary, i.e.\ an adversary that is motivated by financial gain. As such, our adversary is not interested in provoking the opponent to lose a monetary amount during the protocol execution as it would entail a loss for the adversary as well.

\begin{description}
    \item[Malicious Seller] A malicious seller might stop the transmission halfway causing the buyer to receive an incomplete file which is worthless (e.g. compressed file). A malicious seller may also dishonestly report the buyer (e.g.\ pretending not to receive the acknowledgement from the buyer) resulting in the possibility that the smart contract affirms misbehaviour of the buyer. Besides, a malicious seller can try to blackmail the buyer during the execution of the protocol for additional financial reward. We analyze the irrationality of extortion attacks in Section~\ref{attack_analysis}. In the P2P exchange, different sellers can collude with each other (e.g. interrupting one exchange thread, the buyer then can't receive the complete file).
    \item[Malicious Buyer] As the transfer process of \name{} is sequential, a malicious buyer may operate similar attacks (i.e.\ refusing to acknowledge a file chunk so that the seller can't prove the delivery of a chunk, reporting dishonestly and extortion attacks).
\end{description}






\section{\name~Details}
\label{sec:datails}
In the following we outline \name's details.
\subsection{Formal Overview}
We first depict \name in detail for the single seller case. We consider the buyer $B$ is willing to purchase the file $F$ with the hash $H(F)$ from the seller $S$ who claims to possess it. We present in Algorithm~\ref{alg:protocol_pseudocode} the pseudocode of the \name{} protocol which the seller and the buyer iteratively execute for each chunk. \name{} supports the following events:
\begin{itemize}
    \item $T_k$ corresponds to $S$ transmitting the file chunk~$c_k$.
    \item $A_k$ is the acknowledgment of file chunk $c_k$ by $B$.
    \item If one participant does not follow the honest behavior, then the other party ($x$) can report the other's misbehavior to the smart contract for chunk $c_k$ which we denote $R^x_k$.
    \item $P_k^x$ signifies that one party (x) proves to the smart contract for chunk $c_k$ to resolve a dispute.
    \item Not proving a chunk $\overline{P_k^x}$ equals to being offline ($\bot^x$).
\end{itemize}
Notations adopted in this paper are given in Table \ref{tb:notation}.

\begin{algorithm}[htb]\label{alg:protocol_pseudocode}\small
 \KwData{$Pub_B, Pub_S, F_p, D_B, D_S, ...$}
 \KwResult{$B$ gets file F, $S$ is paid $F_p$}
 initialization: set up of a secure channel\;
 \For{$k = n..1$}{
  \textit{$S$ sends chunk $c_k$ and intermediate hash $H_{k-1}^c$\;
  $B$ verifies that $f(f(...f(H_{k-1}^c,M_{CL}^k)...),M_1^k) = H_k^c$\;}
  \eIf{Integrity check passes}{
   \textit{$B$ signs acknowledgement $A_k$ and sends it to $S$\;
   $S$ verifies the validity of $A_k$\;}
    \If{$A_k$ is NOT valid}{
     \textit{$S$ reports to the smart contract that $B$ misbehaves\;}
    }
   }{
   \textit{$B$ reports to the smart contract that $S$ misbehaves\;}
  }
 }
\textit{$S$ finalizes the exchange to withdraw deposits\;} 
\caption{Outline of \name{}'s transfer protocol. Dispute mediation is discussed in Section~\ref{disputemediation}. Note that because of Merkle-Damgard's chaining property, the iterative hashing of the compression function $f$ has to be done in reverse order.}
\end{algorithm}

\begin{table}[!htb]\centering\footnotesize
\caption{Notations adopted for this work.}\label{tb:notation}
\begin{tabular}{ll}
\toprule
Notation & Description \\
\midrule
$F $     & File to be sold\\
$M_i$    & File block $i$\\
$H_i$    & Intermediate hash of $M_i$\\
$h(M)$   & Intermediate hash of block $M$\\
$c_k$    & File chunk $k$\\
$H_k^c$  & Intermediate hash of $c_k$\\
$M_i^c$  & The $i_{th}$ block in $c_k$\\
$s_q$    & File segment $ q$\\
$H_q^s$  & Intermediate hash of $s_q$\\  
$S$      & Seller of the file $F$\\
$B$      & Buyer of the file $F$ \\
$D_S$    & Deposit of the seller \\
$D_B$    & Deposit of the buyer\\
$F_v$    & File value for the buyer\\
$F_p$    & File price paid by the buyer\\
$F_c$    & File selling costs for the seller\\
$T_k$    & Seller successfully transmitted file chunk $k$\\
$A_k$    & Buyer signed ack. of file chunk $k$\\
$R^x_k$  & $x$ reports the other party faulty at chunk $k$\\
$P_k^x$  & $x$ proves to settle a dispute\\
$\bot^x$ & Party $x$ is considered offline\\
$g^x(k)$ & Rel. utility for $x$ after $k$ transferred chunks\\
\bottomrule
\end{tabular}
\end{table}


\subsubsection{Setup Phase}
In the setup phase, $B$ should accept the initial parameters set by $S$ in order to proceed. $B$ and $S$ should agree on the monetary file price $F_p$ and their respective deposits accordingly (cf.\ Section~\ref{sec:nashequilibrium}), which are required in case of dispute mediation. We assume that $B$ and $S$ have their respective public keys $Pub_B$ and $Pub_S$, which correspond to their respective blockchain addresses. The file $F$ is divided into $m$ fixed-size ($BLOCK\_SIZE$) blocks for the hash calculation. $BLOCK\_SIZE$ is determined by the chosen cryptographic hash function (e.g. 64 bytes for SHA-256). The file is transmitted in $n$ chunks. Each chunk consists of a fixed number ($CHUNK\_LEN$ abbr. $CL$) of blocks (cf.\ Equation~\ref{eq:filechunk}), s.t.\ $n=\lceil m/CL \rceil$.
\begin{equation}\label{eq:filechunk}
\begin{split}
  &c_k = \left\{
    \begin{array}{@{}l@{\thinspace}l}
       M_1\ ||\ ...\ ||\ M_{m-(n-k)CL} &\text{,\ if } k=1 \\
       M_{m-(n-k+1)CL+1}\ ||\ ...\ ||\ M_{m-(n-k)CL} &\text{,\ otherwise} \\
    \end{array}
   \right.\\
\end{split}
\end{equation}

Note the length of $c_1$ is not necessarily equal to $CL$. For the convenience of the reader but without loss of generality, we assume $c_1$ contains $CL$ blocks in the rest of this paper.  We denote the $i_{th}$ block of chunk $c_k$ as $M_i^k$ and its intermediate hash as  $h(M_i^k)$. The intermediate hash of $c_k$ is same as the hash of its last block, i.e.\ $H_k^c = h(M_{CL}^k)$. To make sure that the protocol terminates, we define as $MAX\_TIMEOUT$ the maximum time that a participant has to respond before being considered offline, which we set by default to 24 hours. As a permissionless blockchain acts as a timestamping service, we can choose to either use the block height or the block timestamp (UNIX timestamps). Summarizing, in the setup phase $B$ and $S$ should agree on the file price $F_p$, deposits $D_B$, $D_S$, $BLOCK\_SIZE$, $CHUNK\_LEN$ and $MAX\_TIMEOUT$ and deposit to the smart contract respectively start the transmission.

\subsubsection{File Transfer Phase}
We assume that $S$ and $B$ are now capable of connecting off-chain. $S$ starts by sending $B$ the chunk $c_n$ and the intermediate hash $H_{n-1}^c$ (cf. Figure~\ref{fig:merkledamgard}).

After having received $c_n$, the buyer verifies that $f(f(...f(H_{n-1}^c,M_{CL}^n)...),M_1^n) = H(F)$. Note that $f(\cdot)$ corresponds to the one-way compression function, while $H(\cdot)$ corresponds to the iterated compression function over the whole input. 
If the hash is correct, $B$ signs an acknowledgment, $A_n$, and sends it to $S$. The acknowledgment states, that the valid file chunk $c_n$ was received by the buyer. 
$S$ then continues the protocol by giving $B$ the next file chunk $c_{n-1}$ and the intermediate hash $H_{n-2}^c$, to which $B$ answers by signing the acknowledgment $A_{n-1}$.

For every $k \in \{1,2,...,n\}$, $S$ sends $B$ the file chunk $c_{k}$ and the intermediate hash $H_{k-1}^c$, to which $B$ responds by signing $A_{k}$. By the definition of the Merkle-Damg\aa rd construction, $H^c_0 = H_0 = IV$, the initialization vector.
\subsubsection{Dispute Phase}
\label{disputemediation}
Both, buyer and seller are able to report the misbehaviour of the opponent and we consider non-responsiveness to be malicious. Both are therefore incentivized to respond during the execution of the protocol (within $MAX\_TIMEOUT$). In the following paragraphs, we discuss how to handle the misbehavior of either the seller or the buyer.

\paragraph{Misbehavior of the seller}
If $S$ interrupts the protocol or sends an incorrect file chunk at any point during the file transfer phase, $B$ can notify the smart contract that $S$ behaved incorrectly, noted $R^B_k$. $S$ then has $MAX\_TIMEOUT$ to prove the possession of the last block of $c_k$ (i.e.\ $M_{CL}^k$). \name offers two proof methods: First, $S$ publishes $M_{CL}^k$ and the digest (i.e.\ $h(M_{CL-1}^k)$) directly to the smart contract; then, the smart contract verifies that $f(h(M_{CL-1}^k),M_{CL}^k) = H_k^c$. Second, $S$ sends the zkSNARKs proofs to the smart contract without revealing any content of $F$, which is more privacy-preserving but consumes a higher on-chain cost. $S$ can choose the appropriate solution according to the value of $M_{CL}^k$. Depending on whether $S$ provides a valid proof, the smart contract assigns different payoffs to $S$ and $B$. We refer the reader to Section~\ref{payoffphase} for the exact payoff calculation.

\paragraph{Misbehavior of the buyer}
$S$ can report $B$ at any moment if the buyer acknowledges a chunk $c_k$ without having received it\footnote{This may look counter-intuitive, but this functionality enables us to remove the possibility of extortion attacks from the rational seller, as explained in section \ref{sec:nashequilibrium}}, noted $R^S_k$. The smart contract then allows $B$ within $MAX\_TIMEOUT$ to prove the knowledge of $h(M_{CL-1}^k)$ and $M_{CL}^k$, s.t. $f(h(M_{CL-1}^k),M_{CL}^k) = H_k^c$ (similar to the case of seller's misbehavior). $S$ can also report $B$ for not responding an acknowledgement $A_k$ by submitting the last received acknowledgement $A_{k+1}$.

\subsubsection{Payoff Phase}
\label{payoffphase}
Supposing that both parties behave honestly, once $S$ receives the acknowledgment of $A_1$ from $B$, $S$ can submit it to the smart contract to terminate the protocol. In this case, $S$ receives both the file price $F_p$ and the deposit $D_S$, while $B$ receives the deposit $D_B$ and has received the file $F$ through the off-chain communication channel. Depending on the protocol execution steps, for example in case of disputes, we define the particular payouts for the seller and the buyer. We distinguish between the payout defined by the smart contract (i.e., the monetary amount that the smart contract pays) and the utility function (i.e., the utility of the seller and buyer for each chunk of the file). For the ease of understanding, we provide a visual representation in Figure~\ref{fig:payouts}.

\begin{figure*}[!htb]
\centering
\includegraphics[width=0.55\textwidth]{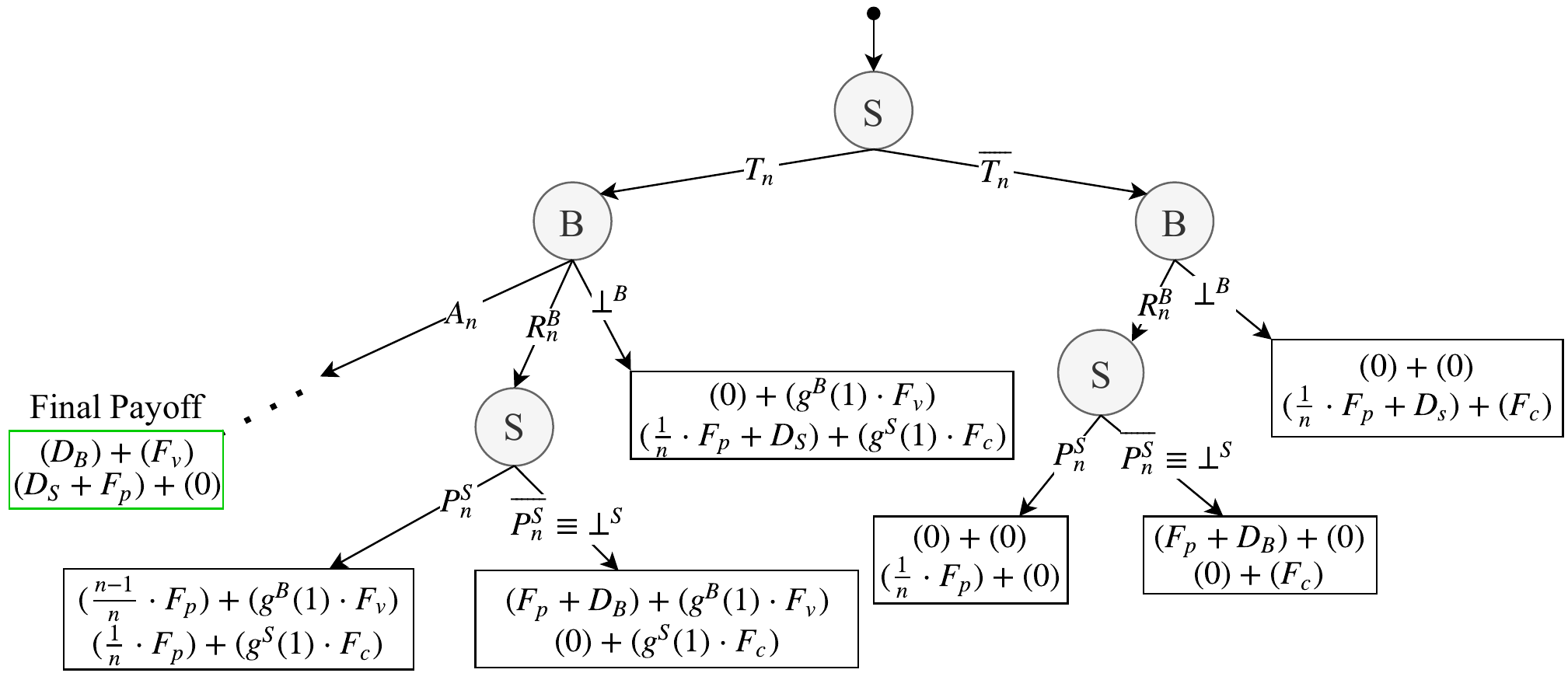}
\caption{Payoff values for the seller and the buyer, depending on the protocol transition steps for two file chunks. $S$ stands for the seller, $B$ stands for the buyer. For each payoff, we distinguish between the payoff by the smart contract (left side) and the utility payoff (right side). For each payoff, the upper payoff concerns the buyer, while the lower concerns the seller. If one participant does not follow the protocol (i.e does not send a chunk, does not acknowledge or goes offline), then the other ($x$) can initiate a dispute for chunk $k$, denoted $R^x_k$; it is then up to the first participant to resolve the dispute. To simplify the presentation, we omitted some sub-cases where a participant becomes offline and also the branch in which buyer acknowledges the receipt of the file chunk in case of not receiving it.} \label{fig:payouts}
\end{figure*}

\subsection{Extension to a Multi-Party File Exchange}
In the P2P exchange setup, files are split into $z$ segments. We extended the fingerprint of a file to the list of all segments' hashes (i.e.\ $\{H_1^s||H_2^s||...||H_z^s\}$). Similar to file chunks, we take the intermediate hash of the last block in segment $s_q$ as the segment hash $H_q^s$. $H_q^s$ is used as the target hash for $s_q$ in the exchange, as well as the initial vector for $s_{q+1}$, so that the buyer can verify the integrity of each segment. We can easily find $H_z^s=H(F)$. In the P2P setup, the smart contract needs to store additional parameters, namely the hash list and the segment size. Based on these, segments can be exchanged separately and concurrently --- segments are transmitted in chunks; the dispute mediation and the payoffs of all segments are independent.

We consider the buyer $B$ is willing to purchase a $z$-segment file $F$ with fingerprint $\{H_1^s||H_2^s||...||H_z^s\}$ from $z$ sellers $S_1, S_2, ..., S_z$. We assume all seller are sharing $F$ at the same price (i.e.\ $F_p/z$ per segment) and deposit (i.e.\ $D_S/z$ per segment). The P2P extension is executed as follows. 
\begin{description}
    \item[Serving] {Each seller starts to share $F$ by sending a transaction with the required parameters (i.e. the fingerprint etc.) and enough deposits which at least exceeds $D_S/z$. The amount of deposits determines the number of segments that can be exchanged concurrently.}
    \item[Handshake] {$B$ is required to prove the willingness of selling a specific chunk of each seller before initiating the exchanges. $B$ thus requests segment $s_q$ from seller $S_q$ with an off-chain message (cf. Equation~\ref{equ:requestmessage}).
    \begin{equation}\label{equ:requestmessage}
        \begin{split}
            r_{s_q} = \{\ &\mbox{buyer}:Pub_B,\  \mbox{seller}:Pub_{S_q},\\ &\mbox{fingerprint}:H_1^s||H_2^s||...||H_z^s,\\ &\mbox{segment}:H_q^s\ \}
        \end{split}
    \end{equation}
    The seller $S_q$ then signs the valid request as the confirmation. We utilize the BLS signature scheme which allows $B$ to aggregate $z$ signatures into a short one so that the smart contract can validate it efficiently. Once all the confirmations are received, $B$ sends the combination of the requests with the aggregated signature, the total price and the deposits to the smart contract.}
    \item[Transfer] {All sellers send chunks to $B$ in descending order. $B$ verifies the authenticity of every received chunk.
    \item[Dispute] The exchange threads are considered to be separate which means that any dispute is mediated independently. The dispute payoff scheme of the P2P scheme imitates the one-to-one exchange.}
    \item[Payoff] {After receiving all the segments, $B$ can notify the smart contract to redeem the deposits with an on-chain transaction. The file price is sent to the seller directly and the sellers' deposits are unlocked in the smart contract for recycling.}
\end{description}

The recycling mechanism facilitates the repeatable selling with one on-chain transaction, which reduces the margin cost of sellers. For a P2P trade, only two on-chain transactions are required (i.e.\ one for the buyer to initiate the transaction and one to finalize the transmission). Note that the same file can be priced differently by distinct sellers and we leave such specification for further work.

\section{Security Analysis}\label{sec:securityanalysis}
Similar to Franklin et al.~\cite{franklin1997fair}, we consider \name{} fair if the file exchanged is consistent with the known hash value. Further formalizing on the attacker model in Section~\ref{sec:attackermodels}, we moreover define utility functions in Section~\ref{sec:utilityfunctions} to measure how $S$ and $B$ value $F$ along with the transmission of the file chunks. Under our rational adversary assumption, in any situation, a player takes action to maximize the final economic benefits. An adversary is not willing to perform an attack if the financial return is lower than the payoffs of following the protocol (i.e.\ $F_p + D_S$ for $S$ and $D_B + F_v$ for $B$). We also assume the trade only proceeds s.t.\ $F_c\leq F_p$ and $F_p\leq F_v$.

\subsection{Attack Analysis}\label{sec:attackanalysis}
\label{attack_analysis}
In this section, we analyze considered malicious behaviours.
\subsubsection{$S$ doesn't have any knowledge of $F$ and trades with $B$} Our mechanism guarantees that $S$ loses the deposit performing this attack. Such an attack is therefore irrational.

\subsubsection{$S$ sends chunks out of order} In each round, $B$ only expects the chunk with the target hash value and considers others invalid, even though $S$ sends a chunk that belongs to $F$. Sending chunks out of order is deemed to be malicious and causes a financial loss of $S$.

\subsubsection{$S$ sends all chunks but $c_1$ of a compressed or encrypted file} In case $B$ reports and then $S$ proves, the payoff for $S$ is $\frac{n-1}{n}\cdot F_p + F_c$, whereas the honest payoff for $S$ is $F_p + D_S$. This attack is accordingly ruled out by choosing the seller's deposit, s.t.\ $D_S > F_c - \frac{1}{n} \cdot F_p$. We further discuss deposit determination in Section~\ref{sec:nashequilibrium}.

\subsubsection{$S$ does not send any file chunk $c_k$ or sends a wrong chunk, and $B$ may be forced to acknowledge it in fear of losing deposit $D_B$} The possibility for $B$ to acknowledge an undelivered chunk gives $S$ an incentive not to send some chunk, because $B$ will acknowledge it anyway to retrieve the deposit. We introduce a counter-intuitive mechanism of offering $S$ the option to report that $B$ acknowledged some chunk $c_k$ without receiving it. If $B$ then fails to prove this, $S$ receives the whole sum $D_S + D_B + F_p$ . Therefore, a rational buyer will never acknowledge an undelivered chunk.

\subsubsection{$B$ extorts $x$ from $S$ after the protocol initiated. Otherwise, $B$ will report the dishonesty of $S$} If disagreeing, $S$ will lose the deposit $D_S$ and receives $F_p/n$. Otherwise, $S$ can recover $D_S+F_p-x$ until the protocol is finalized successfully. This attack seems feasible when $x < D_S+\frac{n-1}{n}\cdot F_p$. However, the optimal strategy of $B$ is to blackmail $S$ in the next every round till the end, which means $S$ will lose more if surrendering. The best response of $S$ is hence to revolt rather than to surrender. Given that both $S$ and $B$ are rational and $B$ knows $S$ is rational, $B$ will not initiate this extortion attack because $S$ is not going to agree.

\subsubsection{$B$ does not acknowledge $c_2$, reports not receiving $c_2$ and tries to brute-force search $c_1$, if $c_1$ is low-entropy} We assume that $B$ can find $c_1$ successfully. Nevertheless, once $S$ prove with $M^2_{CL}$, $B$ only recovers $\frac{2}{n}\cdot F_p$ but loses the deposit $D_B$. Therefore, this attack is irrational with a reasonably large $D_B$. Similarly, due to the risk of losing deposits, the attack that $B$ doesn't acknowledge $c_1$ is unreasonable.

\subsection{Utility functions}
\label{sec:utilityfunctions}
Files differ on how the seller and the buyer value them, in particular regarding individual file chunks. This is relevant for calculating optimal deposit sizes for the existence of a subgame perfect Nash Equilibrium. In practice, though, we do not necessarily know how each participant values a particular file; we, therefore, analyze the worst possible case and find participants' deposits in such a way that the honest behavior is always the subgame perfect Nash Equilibrium. We use a special deposit refunding scheme to achieve that the honest behavior of both seller and the buyer is the subgame-perfect equilibrium solution, see \cite{gersbach} for related schemes.  For some files, the utility function of a participant may grow close to linear (e.g.\ image files\footnote{We acknowledge that image files might be compressed in such a way that coarse-grained details are first transmitted.}, movies), while other files are only valuable in their entirety (e.g.\ encrypted/compressed files). The third category of files might be files where the value of the first chunk outweighs the value of the sum of all other chunks heavily. We define $g^x(k)$ to be the relative utility function, for party $x \in \{Buyer, Seller\}$, after $k$ chunks of the file $F$ are sent. Note that the utility function of the buyer, $g^B(k)$, is monotonically increasing with the number of chunks as he gains more information about the file; on the other hand the utility function of the seller, $g^S(k)$, is monotonically decreasing with the number of chunks as he loses exclusivity of the file.
In order to prove that our protocol is incentive-compatible (it is in the participants' best interest to follow the rules), we analyze the arbitrary utility functions of the buyer and the seller in a rational adversary setting.

\subsubsection{Require the entire file}
For encrypted files, the utility function of the buyer is 0 until the chunk $c_1$ is received (cf. Equation~\ref{eq:entirefile1}).
\begin{equation}\label{eq:entirefile1}
  g^{B}(k) = \left\{
    \begin{array}{@{}l@{\thinspace}l}
       1 \text{, if } k=n \\
       0 \text{, otherwise} \\
    \end{array}
   \right.
\end{equation}

The worst case for the buyer is that the seller's utility function is as follows (cf. Equation~\ref{eq:entirefile2}).
\begin{equation}\label{eq:entirefile2}
  g^{S}(k) = \left\{
    \begin{array}{@{}l@{\thinspace}l}
       0 \text{, if } k=n \\
       1 \text{, otherwise} \\
     \end{array}
   \right.
\end{equation}

In that case, the seller wants the exact opposite of the buyer: he wants to send all the chunks except the last one, while the buyer is only interested in the last chunk.
\name{} therefore requires the deposit of the seller to be sufficiently high, to guarantee that the seller will send the entirety of the file.

\subsubsection{Only require the first chunk}
Recall, that the first chunk that the seller sends corresponds to the end of the file (containing padding and the file length). In some cases, the buyer may solely be interested in knowing this first chunk, for example if he only needs to know the length of the file or in some specific situations where the last chunk is particularly valuable (cf. Equation~\ref{eq:firstchunk1}).
\begin{equation}\label{eq:firstchunk1}
  g^{B}(k) = \left\{
    \begin{array}{@{}l@{\thinspace}l}
       1 \text{, if } k \geq 1 \\
       0 \text{, if k = 0} \\
     \end{array}
   \right.
\end{equation}

The worst case for the seller would be that the buyer's utility function is the exact opposite of the seller (cf. Equation~\ref{eq:firstchunk2}), because after receiving just the first chunk, the buyer already has no incentive to continue the transfer. Yet, the smart contract will only reward the seller with $\frac{1}{n} \cdot F_p$. From the perspective of the buyer, the bulk of the value has already been exchanged (that is, the buyer has already received a value of $F_v$, although the seller only received $\frac{1}{n} \cdot F_p$).
\begin{equation}\label{eq:firstchunk2}
  g^{S}(k) = \left\{
    \begin{array}{@{}l@{\thinspace}l}
       0 \text{, if } k \geq 1 \\
       1 \text{, if k = 0} \\
     \end{array}
   \right.
\end{equation}

All the remaining possible utility functions lie in between these two extremes. In the following section, we make use of the observations derived from the behavior of our mechanism to parametrize the buyer's and seller's deposits, $D_B$ and $D_S$, in such a way that honest behavior is a Nash Equilibrium. This implies that it is in both participants' best interest to behave honestly, as the deposits they provide to the smart contract are high enough and the deposit refunding scheme guarantees higher utility in case of honest behavior than in case of misbehaving.

\subsection{Subgame perfect Equilibrium}\label{sec:nashequilibrium}

In this section, we show that honest behavior in \name{} is a subgame perfect Nash equilibrium.

\subsubsection{Preliminaries}

 We denote the game with $G$. There are two players in $G$, seller and buyer. Game $G$ is sequential, players take their decisions one after another. The strategies available to the seller depending on the level of the game are the following: sending the file chunk, not sending the file chunk, prove sending the file chunk, timing out, reporting the wrong acknowledgment of the buyer and not reporting the wrong acknowledgment. The strategies available to the buyer are acknowledgment of the receipt of the file chunk, claiming not receiving the file chunk, timing out and proving the receipt of the file chunk.  To guarantee honest behavior from both participants, mechanism require them to deposit a monetary amount, $D_B$ for the buyer and $D_S$ for the seller. The \name{} protocol sets what those deposits should be, depending on the price $F_p$ agreed upon. Recall that the smart contract has authority over the money sent by the participants; in case of misbehavior, the smart contract acts as an impartial judge for dispute mediation and distributes payoffs according to the protocol.

\subsubsection{Double Auction}
Note that exchanging the file and setting price is an instance of the double auction problem~\cite{myerson}. In the double auction, a single seller is the owner of one item which is traded with a single buyer. Both the seller and the buyer have privately known values for consuming the item, in our case $F_c$ for the seller and $F_v$ for the buyer. These values are unknown to the mechanism. Intuitively, there are three minimal requirements for the "good" mechanism, individual rationality: the participation of the agents is voluntary, meaning that at any point they may leave the mechanism and consume their initial endowments. The second requirement is budget balance: the mechanism is not allowed to subsidize the agents or to make any profits, the latter means that the amount paid by the buyer is completely transferred to seller. The third requirement is incentive compatibility (IC), meaning that reporting true valuations is the best strategy for both seller and the buyer. 
The seminal paper of Myerson and Satterthwaite \cite{myerson} studies the double auction problem. The well-known impossibility result states that in this setting there is no
mechanism which is fully efficient and at the same time is incentive compatible, individually rational and budget balanced. Blumrosen and Dobzinski~\cite{blumrosen} design an approximately efficient mechanism for achieving optimal expected social welfare function,  that is equal to the maximum of the two private values $F_c$ and $F_v$. In other words, the optimal social outcome is achieved when the item goes to the person who values it the most. It is proved that the only mechanism which satisfies the three properties discussed above is a posted price auction, in which the mechanism posts some price, in our setting $F_p$, if $F_c\leq F_p$ and $F_p\leq F_v$ then the trade takes place. Buyer pays $F_p$ to the seller and gets the item. If one of the conditions does not hold, then no trade takes place. The posted price mechanism is clearly individually rational, because the utilities of both players are non-negative in either case. Reporting true values is clearly a dominant strategy for both players, because reports from both sides do not affect their payoffs, in the calculation of the utility only the private value is used. If any of the players changes the true anticipated outcome of the mechanism, it only decreases his utility. Blumrosen and Dobzinski \cite{blumrosen} design the mechanism which achieves a social welfare that is equal to 0.63 times the optimal social welfare, given only the probability distribution on the value $F_v$. On the negative side, Leonardi et al. \cite{leonardi} prove that there is no mechanism which approximates the optimum by more than a factor of $0.74$ in the worst case. Although, providing efficient posted price can also be integrated into our mechanism, at this point we assume that agreeing on the price $F_p$ is done outside \name{}. Over the time, information can be aggregated to refine probability distributions of the valuations and consequently, more efficient posted price mechanisms can be implemented. 

\subsubsection{Finding the right size for deposits}

In this section, we will determine the optimal sizes for the deposits of both players of the game, in order to guarantee that the honest behavior of both players corresponds to the subgame-perfect equilibrium solution. 

\begin{theorem}
\label{subgame_perfectness}
There exist deposits $D_S$ and $D_B$ such that honest behavior of both players corresponds to the subgame-perfect Nash equilibrium of the extensive form game $G$.
\end{theorem}
\begin{proof}
See Appendix~\ref{proof}.
\end{proof}


From the analysis we exclude the case where a (malicious) buyer may have a file already. In this case, a buyer can acknowledge the receipt of the chunk he has not received, and the rationality analysis boils down to the belief of the seller about the state. Note that given the seller believes the buyer does not have a file, he will send each chunk anyway. Our analysis solves this case. It is easy to observe that the linear refunding scheme by the mechanism of the price $F_p$ is arbitrary and serves only the purpose of simplicity and fairness. Attacks in the form of blackmail do not work. This is guaranteed by the assumptions of rationality and complete knowledge. If any of the parties blackmails the other, asking to collaborate or, otherwise he will make the other party lose deposit, the other party can ignore this kind of threat. Since the other party knows the blackmailing party is rational, the threat will never be realized. One chunk of the file is displayed on blockchain, in case of dispute. This can, in principle, be exploited by malicious buyers, to reveal the file chunk by chunk. On the other hand, since for each chunk they will lose a complete deposit, this is not a rational way to behave. 

\subsection {Secure Multi-Party Exchange}
\label{sec:securemultipartyexchange}
As file segments are exchanged independently, multi-party exchanges are secure as long as all segments are exchanged securely. We prove in Appendix~\ref{proof} that honest behaviour corresponds to a subgame perfect Nash equilibrium in the single seller case when both $D_S$ and $D_B$ are at least $F_p$. Similarly, in the multi-party case, the security of a single thread relies on the conditions that \textbf{(1)} the buyer's deposit for a segment is not smaller than the segment price (i.e.\ $D_B/z \geq F_p/z$); \textbf{(2)} the selling cost of a segment is not larger than the seller's deposit for a segment. We reasonably conclude that the files where value is evenly distributed (e.g.\ images and movies) meet both conditions if $D_S \geq F_p$ and $D_B \geq F_p$. However, for the value-concentrated files, the selling cost of a single segment varies from $0$ up to $F_p$. Therefore, $D_S$ is required to be increased to at least $z\cdot F_p$ in order to preserve the security properties in every exchanging thread.

\section{Implementation and Evaluation}\label{sec:evaluation}
\name{} is a symbiosis of an off- and on-chain component: an Ethereum smart contract serves for dispute mediation; and, an off-chain client is responsible for the actual file transfer (note that the off-chain client needs to monitor the smart contract state). We implement the off-chain component as a browser application written in Typescript ($7\,500$ LOC), while the Ethereum component is a Solidity smart contract ($728$ LOC). The seller and buyer clients first establish a WebRTC connection. Once the connection is established, the protocol as described in Algorithm \ref{alg:protocol_pseudocode} is executed. Note that both clients are required to be synchronized with the blockchain. In consideration of the weakness of SHA-1~\cite{stevens2017first}, we chose SHA-256 as the cryptographic compression function. We implement the SHA-256 hash function in the smart contract and the clients, such that we can execute the Merkle-Damg\aa rd construction on individual file blocks. We also implement a zkSNARK contract for the privacy-preserving possession verification.

\subsubsection{On-chain Costs}
Given no dispute, the on-chain operations of starting a sale, purchasing a file, starting a dispute and finalizing an exchange, cost below $170\,000$ gas (\$$0.09$ at a gas price of $3$ GWei\footnote{At the time of writing, the US dollar to Ether exchange rate is around \$$200$/Ether.}). In \name{}, the most expensive transaction is to resolve a dispute by verifying the possession of a block of 64 bytes. A cleartext verification by calculating SHA-256 hash directly costs about $350\,000$ gas (\$$0.22$), whereas the zkSNARK scheme costs $1\,658\,000$ gas (\$$1.01$). The cost is constant regardless of the file size or chunk size. In a P2P setting, verifying a BLS signature costs approximately $365\,000$ gas (\$$0.23$)~\cite{kfichter69:online}. We estimate the cost of purchasing a file from $10$ sellers to be below $450\,000$ gas (\$$0.28$).

Assuming a PoW blockchain supports ten transactions per second, this amounts to $864\,000$ transactions per day. An estimated upper limit of \name's throughput is thus $432\,000$ P2P downloads per day. This holds assuming the absence of disputes and assuming that a seller is ``re-using'' its collateral for $F$. Note that the measured number of daily downloads via BitTorrent between Dec’03 to Jan’04, is in the range of $237\,500$ to $576\,500$~\cite{pouwelse2005bittorrent}.

\subsubsection{Off-chain Costs}
Transmitting a file of $n$ chunk requires $n$ rounds of interactions between the seller and buyer. The chunk length can be adjusted flexibly without increasing any on-chain cost (e.g.\ a file of $256$MB transmitted in chunks of $256$KB needs $1024$ rounds of interactions which would be reduced to $512$ if increase the chunk size to $512$KB).

\subsubsection{Zero Knowledge Proof Costs}
The zkSnark circuit to verify the possession of a 64-byte block contains $29\,339$ constraints. Generating a zkSnark proof, which has a constant size of $2\,294$ bits, costs $1.585$ seconds on a virtual machine with $16\times 2$GHz CPUs.

\section{Related Work}\label{sec:relatedwork}
Fair exchange has been extensively studied regarding fair multi-party computation \cite{bentov2014use,andrychowicz2014fair,andrychowicz2014secure,banasik2016efficient,bao1999multi,kilincc2015optimally,lysyanskaya2006rationality}, fair multi-party exchange \cite{franklin1997fair,asokan1998optimistic,park2003constructing,dodis2003breaking,zhou2000some,kupccu2010usable,garay2003timed,brickell1987gradual}, contract signing \cite{ateniese1999efficient,gauthierdickey2014secure,bao1998efficient,boldyreva2003threshold} and proof of ownership \cite{halevi2011proofs,di2012boosting}. 
The impossibility of fair exchange without a trusted third party has been proven in \cite{pagnia1999impossibility}. With the recent decade development of the blockchain technology, smart contracts on permissionless blockchains are now generally considered as trusted third parties, which makes fair exchange between multiple parties feasible.
In the following, we mainly discuss related solutions based on permissionless blockchains. NashX~\cite{nashx} attempts to address the fair exchange problem leveraging Bitcoin and the idea of Nash equilibria. Before initiating a trade, both parities deposit some amount of Bitcoin that is more valuable than the transaction amount. Any default may cause the destruction of all the deposits.
This solution is not resistant against extortion attacks, even against a rational adversary, is not automatic and does not provide integrity verification.
FairSwap by Dziembowski et al.~\cite{dziembowski2018fairswap} allows secure digital content exchange through sending ciphertexts off-chain and revealing the key on-chain. FairSwap relies on a judge smart contract to arbitrate the misbehavior of the sender. OptiSwap~\cite{eckeyoptiswap} further extends FairSwap that reduces the overheads in communication and computation. 

\section{Conclusion}\label{sec:conclusion}

In this work, we present a practical and efficient protocol, which allows a buyer to purchase a file from one or multiple sellers. The buyer is assumed to know the cryptographic hash of the file, which allows us to perform an automated dispute mediation through a blockchain based smart contract.
We show that honest behavior in \name{} is a Nash equilibrium under a rational adversary. Under honest execution of the protocol, our implementation only requires three on-chain transactions for the exchange of any file size. In the case of disputes, we require up to two additional on-chain transactions. Contrary to previous work, \name{} has an inherent integrity verification, for every exchanged file chunk. Our evaluations shows that the protocol is practical.

\section*{Acknowledgment}
This work is partially supported by the Hasler Foundation.


\bibliographystyle{IEEEtran}
\bibliography{references}
\appendix
\section{Appendix}

\subsection{Proof of Theorem~\ref{subgame_perfectness}}
\label{proof}
\begin{proof}
In Section~\ref{sec:attackanalysis}, we state that the buyer will never acknowledge a file chunk that hasn't been received, the reason being that acknowledging gives the seller the chance to claim the whole sum, and for the seller it is always strictly dominating strategy to do so. 
We disregard this part of the sequential game $G$ completely. This saves us from considering a game tree which consists of exponentially many states, so we instead deal with a tree of linear size.


We use backwards induction. The claim of the induction is that in each subtree rooted at $S_k$, honest behavior is the subgame perfect equilibrium, i.e. the seller sends chunk at each state and buyer acknowledges it. Final payments from the smart contract in case of honest behavior are $D_B$ for the buyer and $D_S+F_p$ for the seller. 

In the subtree rooted at the node $B_{k}''$, which corresponds to the state that the seller did not send the $k$-th chunk, buyer will always choose to report, $R_{n-k+1}^{B}$, and because of this the seller ends up with the utility 
$$\text{max}\{\frac{k}{n}F_p+g^S(k)F_c, g^S(k-1)F_c\}\leq \frac{k}{n}F_p+F_c,$$
on the other hand, the final utility for the seller in the honest behavior of both players  is equal to $F_p+D_S$, by backwards induction assumption. This implies that if $D_S$ is at least $F_p$ then $F_p+D_S\geq \frac{k}{n}F_p+F_c$, because the right-hand side is maximized when $k=n$ and we also know that $F_p\geq F_c$ by definition, otherwise the trade does not take place at all.

For the buyer an analogous argument works; in the state $B_{k}'$, where the seller already sent the $k$-th chunk, the buyer's payoff in case of playing $R_{n-k+1}^B$ is equal to $$\frac{n-k}{n}F_p+g^B(k)F_v,$$ while by the backwards induction hypothesis, if he plays $A_{n-k+1}$, his utility is equal to $D_B+F_v$. If we take $D_B\geq F_p$, then at each node $B_k'$, the buyer will play honestly, because $D_B+F_v\geq \frac{n-k}{n}F_p+g^B(k)F_v$ for any $k$. Therefore, the honest behavior corresponds to the subgame-perfect Nash equilibrium of the game $G$ with deposits of size $D_S=F_p$ and $D_B=F_p$. 

\end{proof}


For the exposition purposes we do not consider all cases as those cases are treated analogously and the same result can be derived.

\end{document}